\documentclass[letterpaper, 10 pt, conference]{ieeeconf}  % Comment this line out if you need a4paper

\IEEEoverridecommandlockouts                              % This command is only needed if 
\usepackage{graphics} % for pdf, bitmapped graphics files
\usepackage{epsfig} % for postscript graphics files
\usepackage{times} % assumes new font selection scheme installed
\usepackage{amsmath} % assumes amsmath package installed
\usepackage{amssymb}  % assumes amsmath package installed
\usepackage{mymacros} % Macros by Matthias, defiled by Tamas
\usepackage{url} % For code availability statement
\usepackage{tikz}
\usepackage{mathtools}
\mathtoolsset{showonlyrefs}

\usepackage{algorithm}
\usepackage{algpseudocode}

\usepackage{bm}
\usepackage{xcolor}
\newtheorem{theorem}{Theorem}

\newtheorem{remark}{Remark}

\renewcommand{\tilde}{\widetilde}
\renewcommand{\hat}{\widehat}

\renewcommand{\ba}{\mathbf{a}}
\renewcommand{\bz}{\mathbf{z}}

\newcommand{\res}{\mathfrak{r}}
\newcommand{\tres}{\tilde{\res}}

\newcommand{\fomspace}{{\frak{G}_\xi^n}}
\newcommand{\romspace}{{\frak{G}_\xi^r}}

\newcommand{\blambda}{\boldsymbol{\lambda}}

\title{\LARGE \bf
$H_2$-optimal model order reduction using hyperbolic geometry
}

\author{Andrea Angino$^{1}$, Tam\'as D\'ozsa$^{1}$, and Matthias Voigt$^{1}$ % <-this % stops a space
\thanks{The research of AA was funded in part by the Swiss National
Science Foundation (SNSF) grant No. 224943. The research of TD has received funding from the Swiss
Government Excellence Scholarship No. 2025.0057. AA and TD have contributed equally to this work and should both be considered first authors.}% <-this % stops a space
\thanks{$^{1}$Andrea Angino, and Matthias Voigt are with UniDistance Suisse, Faculty of Mathematics and Computer Science, Brig-VS, Switzerland, Tam\'as D\'ozsa is with King Abdullah University of Science and Technology, Computer, Electrical and Mathematical Sciences and Engineering Division, Thuwal, Kingdom of Saudi Arabia
        {{\tt andrea.angino@unidistance.ch, tamasdzs@gmail.com, \newline matthias.voigt@fernuni.ch}%
}}}

\begin{document}

%%%%%%%%%%%%%%%

\onecolumn
\begin{titlepage}
\thispagestyle{empty}

\begin{center}

\vspace*{2.5cm}

{\LARGE\bfseries
$H_2$-optimal model order reduction using hyperbolic geometry
\par}

\vspace{1.2cm}

{\large
Andrea Angino, Tam\'as D\'ozsa, Matthias Voigt
\par}

\vspace{2cm}

{\large\bfseries Author Accepted Manuscript}

\vspace{0.7cm}

\begin{minipage}{0.82\textwidth}
\centering

This document is the accepted version of the manuscript and is not the
final published version.

\vspace{0.8cm}

This paper has been accepted for presentation and publication in the
proceedings of the 65th IEEE Conference on Decision and Control
(CDC 2026), Honolulu, Hawaii, USA, December 15--18, 2026.

\vspace{0.8cm}

The final published version, including the DOI and page numbers,
is not yet available. This cover page will be updated with the full
bibliographic reference and a link to the final published article
once the proceedings become available.

\vspace{0.8cm}

Until then, please refer to this work as:

\medskip

A. Angino, T. D\'ozsa, and M. Voigt,
``$H_2$-optimal model order reduction using hyperbolic geometry,''
accepted for presentation at the 65th IEEE Conference on Decision
and Control (CDC 2026), Honolulu, Hawaii, USA, 2026.

\vspace{1.2cm}

\textcopyright{} 2026 IEEE.

\end{minipage}

\vfill

\end{center}
\end{titlepage}

%ERRATUM:
\clearpage
\thispagestyle{empty}
\vspace*{2.5cm}
\begin{center}
{\Large\bfseries Erratum}

%\vspace{0.4cm}
%{\normalsize\itshape
%$H_2$-optimal model order %reduction using hyperbolic geometry}
\end{center}

\vspace{1.0cm}

In Theorems~\ref{thm:MLcond} and~\ref{thm:thmnewconds},
the reduced transfer function $\bG_r$ should be assumed
to have exact order $r$. Under the assumption of distinct
simple poles, this is equivalent to requiring
$\tres_j \neq 0$ for all $j=1,\ldots,r$.

This assumption is needed to apply the
Meier--Luenberger optimality conditions at all
$r$ mirror-image poles. Furthermore, it excludes
degenerate ROM representations with vanishing residues,
ensuring that all $r$ prescribed mirror-image poles
correspond to actual poles of $\bG_r$.
% In Theorem~\ref{thm:thmnewconds}, the reduced-order transfer function
% $\bG_r$ is additionally assumed to be of \emph{exact order $r$}.
% Under the assumption of pairwise distinct simple poles, this amounts to
% requiring that all residues in its pole--residue representation are
% nonzero, i.e.,
% \[
% \boxed{\tres_j \neq 0 \qquad (j=1,\ldots,r).}
% \]
% This nondegeneracy assumption ensures that the $r$ mirror-image pole
% parameters correspond to actual poles of the reduced-order model, so
% that the Meier--Luenberger conditions used in the proof apply at all
% $r$ points.

\vfill
\noindent\textit{The accepted manuscript follows unchanged.}
\clearpage

\clearpage
\twocolumn

%%%%%%%%%%%
\maketitle
\thispagestyle{empty}
\pagestyle{empty}

\begin{abstract}
We revisit the problem of approximating a linear time-invariant LTI model by a reduced-order model (ROM) in an $H_2$-optimal sense. A widely used method for computing such ROMs is the iterative rational Krylov algorithm (IRKA). However, in practice, the stability of the resulting ROM is not generally guaranteed. In this work, we develop a novel framework that can be used to obtain $H_2$-optimal models of single-input single-output LTI systems while enforcing stability of the ROM. We use Malmquist-Takenaka orthogonal expansions in the corresponding $H_2$-spaces and their connection to hyperbolic geometry to develop an equilibrium function. We show that the poles of an $H_2$-optimal ROM transfer function are defined by the zeros of this equilibrium function. The resulting nonlinear root-finding problem is solved using a Newton-type method formulated on the products of Poincar\'e disk manifolds, where updates are performed along hyperbolic geodesics ensuring that the stability constraint is satisfied at every iteration. We present numerical experiments on benchmark systems for both discrete- and continuous-time LTI MOR problems, including examples where IRKA produces unstable reduced models while the proposed method preserves stability throughout the iteration.
\end{abstract}

\section{Introduction}

Model order reduction (MOR) of linear time-invariant (LTI) systems is important for many control problems in engineering~\cite{djukic2012dynamic}. State-of-the-art MOR methods develop rigorous error bounds for the reduced-order model (ROM) (e.g., balanced truncation~\cite{gugercin2004survey}) or ensure that the ROM satisfies certain optimality conditions~\cite{gugercin2008h_2}. These optimality conditions and error bounds are usually expressed in the frequency domain and are imposed on the error transfer function between the full-order model (FOM) and the ROM. The methods referenced and developed in this work have proper rational transfer functions that belong to the  Hardy-Hilbert space $H_2(\Omega)$, where $\Omega \subset \C$ is an appropriately chosen open domain of analyticity. This assumption holds for transfer functions of any asymptotically stable LTI system and is therefore applicable in a large number of applications~\cite{szidar, feedback}. 

In this study, we propose new methods to obtain ROMs that satisfy the Meier-Luenberger necessary conditions of $H_2$-optimality~\cite{Meier}. These conditions assert that the transfer function of the ROM can only be an $H_2$-optimal approximant of the transfer function of the FOM, if it is an Hermite interpolant at a set of prescribed points that determine the poles of the ROM's transfer function and are often referred to as mirror image poles~\cite{gugercin2008h_2}. 

The state-of-the-art approach for computing a ROM that satisfies these conditions is the iterative rational Krylov algorithm (IRKA) introduced in~\cite{gugercin2008h_2}. In each iteration, IRKA applies Petrov-Galerkin projections to the FOM's state-space matrices. While the convergence of IRKA is only rigorously proven for state space symmetric systems~\cite{FLAGG2012688}, it is widely used in practice. Empirical evidence suggests that it provides high-quality ROMs for most LTI MOR problems~\cite{gugercin2008h_2}. 

Unfortunately, IRKA lacks formal guarantees on the stability of the recovered ROM. Furthermore, it is often observed that even when the final recovered ROM is asymptotically stable, IRKA provides unstable intermediate ROM candidates throughout the iteration~\cite{RiemIRKA}. This can result in unstable ROMs if IRKA is used with inappropriate early stopping criteria.
In~\cite{RiemIRKA} it is shown that IRKA can be interpreted as a Riemannian gradient descent method on the manifold of stable rational transfer functions embedded into $H_2$~\cite[Theorem 4.1, Theorem 7.2]{RiemIRKA}. As a result,~\cite{RiemIRKA} provides a modification of IRKA that uses a line search mechanism to avoid unstable ROM candidates throughout the optimization. The line search however depends on a step size parameter whose choice can influence the efficiency of the method. 

In this work, we take a different approach to construct ROMs that satisfy the $H_2$-optimality conditions. We develop a novel necessary optimality condition for discrete-time systems. In particular, we show that the mirror image poles of $H_2$-optimal ROMs coincide with the zeros of an equilibrium function. Thus, we can pose the $H_2$-optimal MOR problem as nonlinear zero finding problem. In this way we propose a parameter-free alternative to~\cite{RiemIRKA}. For simplicity, we restrict our investigation to single-input single-output (SISO) systems. We obtain the proposed optimality condition by exploiting the properties of the rational orthogonal Malmquist-Takenaka functions~\cite{soumelidis2017hyperbolic}. We will pursue similar results for the multiple-input multiple output (MIMO) case in a future work using the fact that these functions are well defined for matrix-valued $H_2$ spaces (see, e.g.,~\cite{vandenhof}).

 We propose a modified Newton step update that incorporates the hyperbolic geometry of the unit disk to recover the zeros of the equilibrium function. This modification prevents the optimization from leaving the stability region without tuning or even requiring line-search procedures as suggested in~\cite{RiemIRKA}. Additionally, the proposed framework also allows for the use of iterative schemes other than the proposed Newton-based approach. Our  Newton iteration is parameter-independent and guarantees stability by construction. Using the proposed hyperbolic updates, we construct a computationally efficient hybrid Newton-like algorithm that combines Euclidian and hyperbolic steps.

Finally, we employ a unitary transformation between the $H_2$-spaces containing the transfer functions of discrete- and continuous-time asymptotically stable LTI systems. This allows us to apply the proposed hyperbolic/hybrid Newton iteration to continuous-time systems without the need to obtain a new iterative method. We compare the performance of the proposed methods with IRKA in both continuous- and discrete-time settings using benchmark systems from~\cite{slicot}. 

Below we summarize the key contributions of our work:
\begin{enumerate}
    \item We obtain a novel necessary condition of $H_2$-optimality for asymptotically stable discrete-time SISO LTI ROMs. We show that any $H_2$-optimal ROM is defined by the zeros of an equilibrium function.
    
    \item We develop a stability-preserving Newton-like update rule for solving the resulting zero-finding problem on $\mathbb{D}^r$, where $\D$ is the open unit disk. 
    \item We propose an efficient, parameter-free Newton-like method which recovers the zeros of the equilibrium function, while naturally ensuring that the obtained ROMs are asymptotically stable in every step of the iteration.
    \item We use a unitary mapping between Hardy spaces on the open unit disk and the open right half-plane, enabling the same stability-preserving iteration to be applied to continuous-time SISO LTI systems via a conformal change of variables. 
    \item We demonstrate the proposed methods on benchmark SISO LTI systems in both discrete- and continuous-time, and compare against IRKA in terms of achieved $H_2$-error, convergence behavior, and stability of intermediate iterates.
    \end{enumerate}

The paper is organized as follows. Section~\ref{sec:prelims} reviews the necessary discrete-time LTI MOR background. Section~\ref{sec:novelTheor1} presents the new necessary $H_2$-optimality conditions. In Section~\ref{subsec:rgd}, we introduce novel hyperbolic correction steps for the discrete-time setting. Section~\ref{sec:continuous} extends the framework to continuous-time systems via a unitary transformation between Hardy spaces. Numerical results are given in Section~\ref{sec:exp}, and conclusions are drawn in Section~\ref{sec:conc}.

\section{Model order reduction in $H_2(\D)$}
\label{sec:prelims}

Consider the following single-input single-output LTI FOM in discrete-time, given by the state-space representation
\begin{align}
    \label{eq:ss}
    \begin{split}
        \bx(t+1) &= \bA \bx(t) + \bb \bu(t), \quad \bx(0) = 0 \\
        \by(t) &= \bc^\top \bx(t),
    \end{split}
\end{align}
where $\bA \in \C^{n \times n}$, $\bb,\,\bc \in \C^{n \times 1}$ and $\bx(t) \in \C^{n}$, $\bu(t),\,\by(t) \in \C$ for all $t \in \N := \{0,1,\ldots\}$. Here $\bx(t)$, $\by(t)$, and $\bu(t)$ are usually referred to as the state, output, and input at time $t$. We shall associate the order of a system with $n$. The goal of MOR is to identify a ROM of order $r \ll n$, i.e.,
\begin{align}\label{eq:statROM}
    \begin{split}
        \bx_r(t+1) &= \bA_r \bx_r(t) + \bb_r \bu(t), \quad \bx_r(0) = 0 \\
        \by_r(t) &= \bc_r^\top \bx_r(t),
    \end{split}
\end{align}
where $\bA_r \in \C^{r \times r}$, $\bb_r,\bc_r \in \C^{r \times 1}$, $\bx_r(t) \in \C^{r}$, and $\by_r(t) \in \C$ for $t \in \N$ such that $\by_r \approx \by$ for all inputs $\bu$. 

For many relevant MOR methods~\cite{gugercin2004survey, gugercin2008h_2} the FOM and ROM transfer functions are considered. They are obtained using the $\cZ$-transform, i.e.,
\begin{equation}
    \label{eq:Ztrans}
    \cZ[\bg](z) = \sum_{i=0}^{\infty} g_i z^i \quad (z \in \C, \, \bg = (g_0,g_1,\ldots) \in \ell),
\end{equation}
where $\ell$ denotes the vector space of complex sequences indexed by $\mathbb{N}$. The $\cZ$-transform is well defined, whenever $z \in \D \cup \T$, where
\begin{equation}
    \label{eq:disc}
     \D := \{ z \in \C: |z| < 1 \} \quad \text{and} \quad \T:=\{ z \in \C : |z| = 1 \}
\end{equation}
and 
\begin{equation*}
    \bg \in \ell_2 := \left\{ \bg \in \ell \;:\; \sum_{i=0}^{\infty} |g_i|^2 < \infty \right\}.
\end{equation*}
 Let us assume that $\bu \in \ell_2$ and that the FOM is asymptotically stable, i.e., $|\mu| < 1$ for all eigenvalues $\mu$ of $\bA$. In this case, we also have $\bx \in \ell_2$ and $\by \in \ell_2$ and the $\cZ$-transforms of $\bu$, $\bx$, and $\by$ are all well-defined. 
 
We now derive a transfer function representation for the FOM given in~\eqref{eq:ss} compatible with our $\cZ$-transform convention.
{
Using standard manipulations, the discrete-time transfer function is
\begin{equation*}
    \hat{\bG}(z)
    :=
    \bc^\top \left( \frac{1}{z}\bI_n - \bA \right)^{-1}\bb
    =
    z\,\bc^\top (\bI_n-z\bA)^{-1}\bb
\end{equation*} for $z \in \D \setminus \{0\}$. Let  $\mu_k$ be the eigenvalues of $\bA$. Assume they are simple, then $\hat{\bG}$ admits the partial fraction expansion
\begin{equation}
    \label{eq:freqRep}
    \hat{\bG}(z)=\sum_{k=1}^n \frac{z\,\res_k}{1-\mu_k z},
\end{equation}
where $\res_k \in \C$.
Motivated by~\eqref{eq:freqRep}, we introduce the function
\begin{equation}
    \label{eq:discTrfRaw}
    \bG(z):=\frac{1}{z}\hat\bG(z) = \bc^\top (\bI_n-z\bA)^{-1}\bb
    =
    \sum_{k=1}^n \frac{\res_k}{1-\mu_k z}.
\end{equation}
The latter expression continues analytically to $z=0$ with ${\bG}(0)=\bc^\top \bb$.
Consider the Hardy-Hilbert space \begin{equation}
    \label{eq:H2D}
    H_2(\D) := \left\{ \bG \in \frak{A}(\D) \;:\; \frac{1}{2 \pi} \int_{-\pi}^{\pi} \big| \bG\big(\eu^{\iu \theta}\big) \big|^2 \,\du \theta < \infty \right\},
\end{equation}
where $\frak{A}(\D)$ denotes the set of analytic functions $\D \to \C$.
Since $z \mapsto z$ is an inner function on $\D$, i.e., $|z|=1$ on $\T$, multiplication by $z$ preserves the $H_2(\D)$-norm. In particular, for any reduced model,
\[
\big\|\hat{\bG}-z\bG_r\big\|_{H_2(\D)} = \big\|\bG-\bG_r\big\|_{H_2(\D)},
\]
where  $\bG_r$ denotes the transfer function corresponding to Eq.~\eqref{eq:statROM}.
Therefore, the $H_2$-model reduction problem can be formulated equivalently in terms of $\bG$. By setting \(
\lambda_k := \overline{\mu_k},
\)
we arrive at the representation
\begin{equation}
    \label{eq:discTrf}
    \bG(z) = \sum_{k=1}^n \frac{\res_k}{1-\overline{\lambda_k} z}.
\end{equation}
The parameters $\lambda_k$ in~\eqref{eq:discTrf} will be referred to as mirror image poles of $\bG$. Notice that the poles of $\bG$ are given by
\begin{equation}
    \label{eq:discPoles}
    1 / \overline{\lambda}_k \quad (k=1,\ldots,n).
\end{equation}
Henceforth, we denote by $\bG_r$ the transfer function of the ROM, which matches the structure of~\eqref{eq:discTrf} and is described by the parameters $\tres_j$ and $\tlambda_j \ (j=1,\ldots,r)$.
}

We will assume that $\bG$ is stable, i.e., all its poles are outside $\overline{\D}$ and thus, $\bG \in H_2(\D)$. We shall also assume that $\bG$ has simple poles. In our method, we will also require thar $\bG_r$ satisfies the same conditions. These are standard conditions in interpolatory MOR. While we can represent discrete-time transfer functions with either definition, using~\eqref{eq:discTrf} allows us to exploit strong connections to hyperbolic geometry in Section~\ref{subsec:rgd}, which are later used to construct stability-preserving Newton updates. 

% ***

To formalize the considered MOR problem, we define the kernel function 
\begin{equation}\label{eq:ckerslice}
 \xi: \D \times \D \to \C, \quad \xi(\lambda,z) := \xi_{\lambda}(z) := \frac{1}{1 - \overline{\lambda} z}.
\end{equation}
Define
$\blambda := (\lambda_1,\ldots,\lambda_n)$ and  
\begin{multline}
    \label{eq:ROMspace}
    \fomspace := \left\{ \bG \in \myspan\left\{ \xi_{\lambda_1},\ldots,\xi_{\lambda_n}\right\} : \right. \\ \left. \blambda \in \D^n,\, \lambda_j \neq \lambda_k \text{ for } j \neq k \right\}.
\end{multline}
Note that $\fomspace \subset H_2(\D)$. Then we pose the discrete-time LTI MOR problem as minimizing the $H_2(\D)$-error functional $E : \romspace \to [0, \infty)$ given by
\begin{equation}
    \label{eq:H2errfunc}
    E(\bG_r) := {\| \bG - \bG_r \|}_{H_2(\D)},
\end{equation}
for a prescribed FOM transfer function $\bG \in \fomspace$.

We now recall some necessary conditions that $\bG_r \in \romspace$ needs to satisfy if it is a local minimizer of $E$ in~\eqref{eq:H2errfunc}. 
\begin{theorem}[Meier-Luenberger conditions~\cite{gugercin2008h_2, Meier}]
\label{thm:MLcond}
    Let $\bG \in \fomspace$ be a given FOM transfer function. Suppose that $\bG_r \in \romspace$ is a local minimizer of $E$ in~\eqref{eq:H2errfunc}, then
\begin{equation}
    \label{eq:MLC}
    \bG_r\big(\tlambda_j\big) = \bG\big(\tlambda_j\big), \quad \bG_r'\big(\tlambda_j\big) = \bG'\big(\tlambda_j\big) \quad (j=1,\ldots,r),
\end{equation}
where $\tlambda_j$ denote the mirror image poles of $\bG_r$.
\end{theorem}

The state-of-the-art algorithm for computing a reduced model $\bG_r$ that satisfies Theorem~\ref{thm:MLcond} is IRKA. This method can be interpreted as a Krylov-type fixed point iteration, relying on repeated Petrov-Galerkin projections to update the mirror image poles $\tlambda_j$ and residues $\tres_j$ for $j=1,\ldots,r$.
Whenever IRKA converges to a ROM transfer function $\bG_r \in \romspace$, the resulting ROM satisfies the Meier-Luenberger conditions in Theorem~\ref{thm:MLcond}. However, IRKA does not enforce the constraint $\bG_r \in \romspace$ during the iteration. In particular, intermediate iterates may leave $\romspace$, and even the final iterate may (in principle) fail to belong to $\romspace$. In such cases, the assumptions underlying Theorem~\ref{thm:MLcond} are violated, and the optimality conditions are no longer guaranteed.
In contrast, the approach proposed in this work is constructed so that all iterates remain in $\romspace$ by design. More precisely, we reformulate the MOR problem in~\eqref{eq:H2errfunc} as a zero-finding problem on $\D^r$, and define fixed point iterations that induce a mapping
\(
\romspace \;\to\; \romspace.
\)
As a consequence, every intermediate reduced model is asymptotically stable. In particular, the admissibility condition $\bG_r \in \romspace$ required in Theorem~\ref{thm:MLcond} is satisfied along the entire iteration. The optimality conditions in Theorem~\ref{thm:MLcond} are then recovered at convergence of the proposed iteration.

In~\cite{RiemIRKA}, the authors propose the extension of IRKA with a line search mechanism to prevent the method from reaching unstable regions of the search space. This, however, is parameter-dependent and can be costly from a computational standpoint. In Section~\ref{sec:novelTheor1}, we propose a completely new alternative framework to IRKA that ensures the stability of the ROM and avoids the downsides of the line search-based mechanism. For a deep discussion on IRKA we recommend~\cite{gugercin2008h_2} and~\cite{RiemIRKA}.

Let $\ba = (a_0,a_1,\ldots)$ be a sequence in $\D$ and consider the Malmquist-Takenaka (MT) functions in $H_2(\D)$ defined by
\begin{equation}
    \label{eq:H2DMT}
    \varphi_k^{\ba}(z) := \frac{\sqrt{1 - |a_k|^2}}{1 - \overline{a}_k z} \prod_{j=0}^{k-1} \frac{z - a_j}{1 - \overline{a_j} z} \quad (z \in \D \cup \T).
\end{equation}
In~\eqref{eq:H2DMT}, the terms $\frac{z - a_j}{1 - \overline{a}_j z}$ are called Blaschke factors. These are Möbius transformations defined on special linear group $SL(2)$. They also describe congruence transformations in the Poincar\'e disk model of Bolyai's hyperbolic geometry. These properties make them a key component of the Newton update in Section~\ref{subsec:rgd}. The MT functions in~\eqref{eq:H2DMT} form an orthonormal system with respect to the $H_2(\D)$ inner product 
$$
    {\langle \bF, \bG \rangle}_{H_2(\D)} = \frac{1}{2 \pi} \int_{-\pi}^{\pi} \bF\big(\eu^{\iu \theta }\big) \overline{\bG\big(\eu^{\iu \theta }\big)}\, \du \theta \quad (\bF, \bG \in H_2(\D)).
$$
In fact, the MT system is complete if the generating sequence $\ba$ satisfies the Sz\'asz condition (see, e.g.,~\cite{SzegoProof}). This is always true for $p$-periodic MT systems, with generating sequence
$$
    a_{i} = a_{i + kp}, \quad (i=0,\ldots,p-1, \ k \in \N).
$$
Consider the $p$-th MT-Fourier partial sum of $\bG \in H_2(\D)$,
\begin{equation}
    \label{eq:partsum}
    \cS_p^{\ba}[\bG] = \sum_{k=0}^{p-1} {\langle \bG, \varphi_k^{\ba} \rangle}_{H_2(\D)} \varphi_k^{\ba}.
\end{equation}
This is an orthogonal projection of $\bG$ onto the subspace
\begin{equation}
    \label{eq:modspace}
    \myspan \{ \varphi_0, \ldots, \varphi_{p-1} \} = \myspan \{ \xi_{a_0}, \ldots, \xi_{a_{p-1}} \} \subset H_2(\D),
\end{equation}
where $\xi_{a_k}$ is defined according to~\eqref{eq:ckerslice}. By~\eqref{eq:modspace}, if $a_{k-1} = \lambda_k \ (k=1,\ldots,n)$, we have $\bG = \cS_n^{\ba}[\bG]$ if $\bG \in \fomspace$.

We can express the pointwise error of the ROM directly using periodic MT expansions (see, e.g.,~\cite{soumelidis2017hyperbolic}). Let $\ba \in \D^r$. The $r$-periodic sequence generated by the vector $\ba$ is also denoted by $\ba$ henceforth. Then $\bG_r := \cS_r^{\ba}[\bG]$ satisfies
\begin{multline}\label{eq:discPointWise}
   \bG(z) - \cS_r^{\ba}[\bG](z) \\ = B^{\ba}(z) \left( \sum_{k=1}^n \frac{\res_k}{1 - \overline{\lambda}_k z} \overline{B^{\ba}(\lambda_k)} \right) \quad (z \in \D \cup \T),
\end{multline}
where
\begin{equation}
    \label{eq:Bprod}
    B^{\ba}(z) := \prod_{j=0}^{r-1} \frac{z - a_j}{1 - \overline{a}_j z}
\end{equation}
denotes the $r$-term Blaschke product. It satisfies $B^{\ba}(a_j) = 0$ for $j=0,\ldots,r-1$, hence the projection $\bG_r = \cS_r^{\ba}[\bG]$ interpolates the FOM transfer function $\bG$ at $\{ a_j \}_{j=0}^{r-1} \subset \D$ by~\eqref{eq:discPointWise}.

\section{A novel $H_2(\mathbb{D})$-optimality condition}
\label{sec:novelTheor1}
We pose the $H_2(\D)$-optimal MOR problem from~\eqref{eq:H2errfunc} as a zero finding task. 
\begin{theorem}[A necessary $H_2(\D)$ optimality condition]
\label{thm:thmnewconds}
    Let the FOM transfer function $\bG \in \fomspace$. Consider the \emph{equilibrium function} $F_\bG : \D^r \to \C^r$, where for $\bz = (z_1,\ldots,z_r) \in \D^r$,
    \begin{equation}\label{eq:FG}
\big[F_\bG(\bz)\big]_j:= \sum_{k=1}^n \frac{\res_k}{1 - \overline{\lambda_k} z_j} \overline{B^{\bz}(\lambda_k)} \quad (j=1,\ldots,r),
    \end{equation}
where $B^{\bz}$ is an $r$-term Blaschke product (see~\eqref{eq:Bprod}).
    If an order-$r$ ROM transfer function $\bG_r \in \romspace$ is a local minimizer of $E$ from~\eqref{eq:H2errfunc}, then 
    \begin{equation}
        \label{eq:equibfun}
F_\bG\big(\boldsymbol{\tlambda}\big) = 0,
    \end{equation}
    where $\boldsymbol{\tlambda} =  %denotes the $r$-periodic $\D$ sequence generated by
    \big(\tlambda_j\big)_{j=1}^r \in \D^r$ %that 
    coincides with the mirror image poles of $\bG_r$. %The map $F_\bG$ is referred to as the order-$r$ \emph{equilibrium function} of the FOM $\bG$.
\end{theorem}
\begin{proof}
    %Without loss of generality, we can assume that since $\bG_r$ is a local minimizer of $E$ from~\eqref{eq:H2errfunc}, $\bG_r$ is a projection onto the subspace $\myspan\left\{\xi_{\tlambda_1}, \ldots, \xi_{\tlambda_r}\right\}$, hence 
    Note that $\bG_r = \cS_r^{\boldsymbol{\tlambda}}[\bG]$. By~\eqref{eq:discPointWise}, for any $\boldsymbol{\tlambda} \in \D^r$, since $\bG_r = \cS_r^{\boldsymbol{\tlambda}}[\bG]$, we have 
    \begin{equation}
        \label{eq:interp}
        \bG(\tlambda_j) = \bG_r(\tlambda_j) \quad (j=1,\ldots,r).
    \end{equation}
    By Theorem~\ref{thm:MLcond}, a local minimzer $\bG_r$ also satisfies
    \begin{equation}
        \label{eq:hinterp}
        \bG'(\tlambda_j) = \bG_r'(\tlambda_j) \quad (j=1,\ldots,r).
    \end{equation}
    We show that~\eqref{eq:hinterp} is equivalent to~\eqref{eq:equibfun} under our assumptions. Differentiating both sides of~\eqref{eq:discPointWise} gives
    \begin{multline}
        \label{eq:deriv}
       \bG'(w) - \left( \cS_r^{\bz}[\bG] \right)'(w) \\ =  \left(B^{\bz}\right)'(w) \left( \sum_{k=1}^n \frac{\res_k}{1 - \overline{\lambda}_k w} \overline{B^{\bz}(\lambda_k)} \right) \\ + B^{\bz}(w)  \left( \sum_{k=1}^n \frac{\res_k}{1 - \overline{\lambda}_k w} \overline{B^{\bz}(\lambda_k)}\right)',
    \end{multline}
    where $w \in \D$. Notice that by~\eqref{eq:Bprod}, for $w \in \big\{z_j\big\}_{j=1}^r$ we have $B^{\bz}(w) = 0$. Fixing $\bz = \boldsymbol{\tlambda}$ and choosing $w \in \big\{\tlambda_j\big\}_{j=1}^r$, i.e., the parameters that determine $\bG_r$, it suffices to ensure
    $$
    \left( B^{\boldsymbol{\tlambda}}\right)'(\tlambda_j) \left( \sum_{k=1}^n \frac{\res_k}{1 - \overline{\lambda}_k \tlambda_j} \overline{B^{\tilde{\blambda}}(\lambda_k)} \right) = 0 \quad (j=1,\ldots,r).
    $$
   We note that $\left(B^{\boldsymbol{\tlambda}}\right)'(\tlambda_j) = 0$ is only satisfied if $\tlambda_j$ is a mirror image pole with multiplicity at least $2$. Indeed, since
   $$
        B^{\boldsymbol{\tlambda}}(\tlambda_j) = \frac{P(\tlambda_j)}{Q(\tlambda_j)} = 0
   $$
   for polynomials $P$ and $Q$ of degree $r$ (see~\eqref{eq:Bprod}),
   we have $P(\tlambda_j) = 0$,  However,
   $$
        \left( B^{\boldsymbol{\tlambda}}\right)'(\tlambda_j) = \frac{P{'}(\tlambda_j) Q(\tlambda_j) - P(\tlambda_j)Q{'}(\tlambda_j)}{\left( Q(\tlambda_j) \right)^2} = \frac{P{'}(\tlambda_j)}{Q(\tlambda_j)},
   $$
   and thus $\left( B^{\boldsymbol{\tlambda}}\right)'(\tlambda_j) = 0$ is equivalent to $P'(\tlambda_j) = 0$. Since by assumption, the mirror image poles of $\bG_r$ are simple, we therefore find $\left( B^{\boldsymbol{\tlambda}}\right)'(\tlambda_j) \neq 0$. Thus,  
    \begin{equation*}
        \bG'(\tlambda_j) - \bG_r'(\tlambda_j) = 0 
        \iff \big[F_{\bG}(\boldsymbol{\tlambda})\big]_j = 0 \quad (j=1,\ldots,r).
    \end{equation*}
\end{proof}

\section{Zero finding strategies}
\label{subsec:rgd}
Now we construct a Newton-like iteration, with the aim of obtaining a $\D^r$-sequence $(\bz^{(\nu)})_{\nu=0}^{\infty}$ such that
\begin{equation*}
    \lim_{\nu \to \infty} \bz^{(\nu)} = \boldsymbol{\tlambda},
\end{equation*}
where $\boldsymbol{\tlambda} \in \D^{r}$ satisfies~\eqref{eq:equibfun}. Thus, it characterizes the ROM as $\bG_r = \cS_r^{\boldsymbol{\tlambda}}[\bG] \in \romspace,$ which satisfies the necessary optimality conditions of Theorems~\ref{thm:MLcond} and~\ref{thm:thmnewconds}.
Equation~\eqref{eq:equibfun} defines a system of nonlinear equations in \(
\bz = (z_1,\ldots,z_{r}) \in \mathbb{D}^r
\),
where the dependence on $\bz$ is implicit through the Blaschke product $\overline{B^{\bz}}$.
Hence, determining $\bz$ is a root-finding problem that can be solved with a Newton-like method~\cite{dozsa2023inverses}.

\subsection{Hyperbolic correction of the Newton step}
Consider $F_\bG :\mathbb{D}^r\to\mathbb{C}^r$ from~\eqref{eq:FG}.
Let $\Delta \bz_{\rm N}^{(\nu)}$ denote the Newton correction defined at the $\nu$-th iteration by \begin{equation}\label{eq:NewtonSystem}
    J\big(\bz^{(\nu)}\big) \Delta \bz_\mathrm{N}^{(\nu)} = -F_\bG\big(\bz^{(\nu)}\big),
\end{equation}
where $J(\bz)$ is the Jacobian of $F_{\bG}(\bz)$ (in the Wirtinger sense, see,  e.g.,~\cite{henrici1988applied}). Therefore, a classical Newton update is
\begin{equation}\label{eq:NewtonStep}
\bz^{(\nu+1)} = \bz^{(\nu)} + \Delta \bz_\mathrm{N}^{(\nu)}.    
\end{equation}
Following~\eqref{eq:NewtonStep} can exhibit fast local convergence, but it provides no intrinsic mechanism to preserve the constraint
\(
\bz \in \mathbb{D}^r.
\)
In particular, if the update satisfies $\left|z_j^{(\nu)}+\Delta z_{\textrm{N},j}^{(\nu)}\right|\geq 1$ for some $j$, then $\mathcal{S}_r^{\bz^{(\nu )}}[\bG]\notin \romspace$. In Section~\ref{sec:exp}, we report an example where this behavior occurs, motivating the following modifications to the Newton step.

A common remedy is to impose explicit safeguards in Euclidean space, e.g., by line-search,
\(
\bz^{(\nu+1)}=\bz^{(\nu)} + \eta \Delta \bz_\mathrm{N}^{(\nu)}
\)
with $\eta\in(0,1]$ chosen so that $\bz^{(\nu+1)}\in\mathbb{D}^r$.
However, it is parameter-dependent, may fail to preserve stability if the parameters are not properly tuned, and introduces additional computational overhead. Similar problems manifest in the line search variant of IRKA~\cite{RiemIRKA}.

An alternative approach is to reinterpret the root-finding task in~\eqref{eq:equibfun} as a constrained minimization problem, e.g., 
\begin{equation}
\label{eq:mintask}
\min_{\bz \in \mathbb{D}^r}\frac{1}{2}{\|F_\bG(\bz)\|}_{2}^2,
\end{equation}
or, equivalently, as a sequence of trust-region subproblems in which each step is constrained to remain inside $\mathbb{D}^r$. Nevertheless, this approach can cause $z_j^{(\nu)}\in \T$ for some $j\in \{1,\dots,r\}$ due to projections~\cite{conn2000trust}.

In this work, we instead solve~\eqref{eq:equibfun} while ensuring that the ROM transfer function satisfies $\mathcal{S}_r^{\bz^{(\nu )}}[\bG]\in \romspace$ at each iteration. This is achieved by modifying the Newton step in~\eqref{eq:NewtonStep} so that it becomes compatible with the geometry induced by the pseudo-hyperbolic metric
\begin{equation}
\label{eq:pseudmet}\rho(z_1,z_2)
:=
\left|\frac{z_2 - z_1}{1-\overline{z_1}z_2}\right|=\left| B^{z_1}(z_2) \right|
\quad (z_1,z_2 \in \mathbb{D})
.
\end{equation}
In this view, the unit disk is not merely a convenient constraint set, but carries a natural intrinsic geometry induced by Blaschke factors \eqref{eq:Bprod}.

In the following, we adopt a Riemannian viewpoint. Therefore, we assume that the reader is familiar with basic notions and terminology of Riemannian optimization. For an overview, we recommend~\cite{boumal2023intro,absil2008matrix}.

Consider the Poincar\'e disk model $(\mathbb{D},g_{\mathbb{D}})$, where the open unit disk
\(
\mathbb{D} \subset \mathbb{C}
\)
is endowed with the hyperbolic metric
\begin{equation}
\label{eq:poincare_metric}
g_{\mathbb{D}}(z) = \gamma(z)^2\, |\du z|^2,
\qquad
\gamma(z):=\frac{2}{1-|z|^2},
\end{equation}
where $|\cdot|$ is the standard Euclidean metric.\footnote{The term ``metric'' has been used in way that is standard in Riemannian geometry.} 
Under this metric, the unit disk becomes a complete Riemannian manifold.

A key property of the Poincar\'e geometry is that the boundary $\T$ lies at an infinite hyperbolic distance. Consequently, geodesic-type updates generated on $(\mathbb{D},g_{\mathbb{D}})$ cannot reach the boundary in a finite number of steps.
Let $\oplus$ denote the M\"obius addition on $\mathbb{D}$,
\[
z\oplus w := \frac{z+w}{1+\overline{z}\,w}=B^{-z}(w)
\quad (z,w\in\mathbb{D}).
\]
The exponential map of the Poincar\'e disk  admits the following closed form expression
\begin{equation}
\label{eq:ExpMap}
\operatorname{Exp}_z(q)
=
z \oplus 
\left(
\tanh\!\left(\frac{\gamma(z)\,|q|}{2}\right)\frac{q}{|q|}
\right),
\end{equation}
with the convention $q/|q|:=0$ if $|q|=0$. Thus, starting from $z\in \D$, $\operatorname{Exp}_z(q)$ moves in the direction $q\in T_z\D$ along a geodesic. Here, $T_z \D$ denotes the tangent space of $\D$ at $z$.

The product manifold
\begin{equation*}
    \label{eq:manifM}
    \frak{M}:=(\mathbb{D},g_\D)^r
\end{equation*}
equipped with the product metric inherits the structure of a smooth Riemannian manifold. In practice, this implies that intrinsic update rules on $\frak{M}$ preserve the stability constraint $\bz \in \D^r$.
%\subsection{Geodesic correction of the Newton direction}
Let 
\(
T_{\bz}\frak{M}
\)
 be the tangent space of $\frak{M}$ at $ \bz\in\frak{M}$.
Since $\mathbb{D}$ is an open subset of $\mathbb{C}$, we identify 
\(
T_{\bz}\frak{M} \simeq \mathbb{C}^r.
\)

To account for the hyperbolic geometry we apply a metric-induced scaling to the Newton correction by
%Define the hyperbolically corrected Newton displacement by
\begin{equation}
\label{eq:HyperbolicScaling}
\Delta \bz_\mathrm{H}^{(\nu)}
:=
\operatorname{diag}\!\Big(\gamma\big(z_1^{(\nu)}\big)^{-2},\ldots,\gamma\big(z_{r}^{(\nu)}\big)^{-2}\Big)\,
\Delta \bz_\mathrm{N}^{(\nu)}.
\end{equation}
Equivalently, component-wise we obtain
\(
\big[\Delta \bz_\mathrm{H}^{(\nu)}\big]_j = \frac{\big(1-\big|z_j^{(\nu)}\big|^2\big)^2}{4}\big[\Delta \bz_\mathrm{N}^{(\nu)}\big]_j
\) for $j=1,\ldots,r$.
We realize the corrected Newton update by moving along hyperbolic geodesics, i.e.,
\begin{equation}
\label{eq:GeodesicNewton}
z_j^{(\nu+1)}
=
\operatorname{Exp}_{z_j^{(\nu)}}\!\left(\big[\Delta \bz_\mathrm{H}^{(\nu)}\big]_j\right)
\quad (j=1,\ldots,r).
\end{equation}
On $\frak{M}$,  $\operatorname{Exp}_{\bz}$ is applied component-wise as in~\eqref{eq:GeodesicNewton}.

Thus, the Newton direction $\Delta \bz_\mathrm{N}^{(\nu)}$ is preserved, while the hyperbolic geometry modifies how this displacement is realized on the manifold.
By construction, $\operatorname{Exp}_{z_j}(q)\in\mathbb{D}$ for all $z_j\in\mathbb{D}$ and $q\in T_{z_j}\mathbb{D}$, and the iterates~\eqref{eq:GeodesicNewton} remain in $\mathbb{D}^r$.

\iffalse
\begin{remark}
    Following the Newton update in~\eqref{eq:GeodesicNewton} for finding the zeros of the equilibrium function $F_\bG$ is equivalent to taking a step with Riemannian gradient descent while solving problem~\eqref{eq:mintask} with manifold $\frak{M}$ (see~\eqref{eq:manifM}).
\end{remark}
\fi

\subsection{Hybrid strategy}
Purely geometry-aware updates can be more expensive than Euclidean Newton iterations due to the cost of
evaluating the exponential map.
Moreover, far from the boundary, Euclidean Newton steps are safe and converge rapidly.

For this reason, we employ a hybrid strategy:
\begin{equation}
\label{eq:hybrid}
\bz^{(\nu+1)}
=
\begin{cases}
\bz^{(\nu)} + \Delta \bz_\mathrm{N}^{(\nu)}, &
\text{if } \bz^{(\nu)} + \Delta \bz_\mathrm{N}^{(\nu)} \in \mathbb{D}^r,\\[1mm]
\operatorname{Exp}_{\bz^{(\nu)}}\!\left(\Delta \bz_\mathrm{H}^{(\nu)}\right), &
\text{otherwise}.
\end{cases}
\end{equation}
In other words, we apply the (faster) Euclidean Newton step whenever the candidate iterate remains inside $\mathbb{D}^r$, and switch to the hyperbolically corrected step only when the Newton iterate would leave the disk.
In finite precision arithmetic, we additionally apply a small safety margin to avoid numerical issues near the boundary $\T$.

Since $\bz^{(\nu)}$ must lie in $\mathbb{D}^r$ to ensure stability, the update~\eqref{eq:hybrid}, similarly to~\eqref{eq:GeodesicNewton}, guarantees that \emph{every intermediate reduced model satisfies} $\bG_r^{(\nu)}=\cS_r^{\bz^{(\nu)}}[\bG]\in \romspace$. Finally, we note that using the recovered mirror image poles, we can obtain a state space model of $\bG_r$ using a Galerkin projection method introduced in~\cite[Section 4.1]{gugercin2008h_2}.

\section{Extension to continuous-time systems}\label{sec:continuous}

Often, continuous-time LTI systems are described directly by their transfer functions
\begin{equation*}
    \bH(s) := \bc^\top (s\bI_n - \bA)^{-1}\bb,
\quad
\bH_r(s) := \bc_r^\top (s\bI_r - \bA_r)^{-1}\bb_r,
\end{equation*}
where the matrices $\bA, \bb, \bc$ determine a state space representation of the system similar to~\eqref{eq:ss}, but with the assumption $t \in [0, +\infty)$. Let $\C^+$ denote the complex open right half-plane and 
\begin{equation}
    \label{eq:contKer}
    \zeta: \C^+ \times \C^+ \to \C, \quad \zeta(s, \lambda) := \zeta_{\lambda}(s) := \frac{1}{s + \overline{\lambda}}.
\end{equation}
With this, we can write the partial fraction representations
\begin{equation}
    \label{eq:trf}
    \bH(s) = \sum_{k=1}^n \res_k \zeta_{\lambda_k}(s), \quad \bH_r(s) = \sum_{k=1}^r \tres_k \zeta_{\tlambda_k}(s),
\end{equation}
where $\lambda_k = -\overline{\mu_k}$ and $\mu_k$ $(k=1,\ldots,n)$ denote eigenvalues of $\bA$,  $\tlambda_j = -\overline{\tilde{\mu}_j}$ and $\tilde{\mu}_j$ $(j=1,\ldots,r)$ denote the eigenvalues of $\bA_r$ from the state space representation. 

If both the FOM and ROM are asymptotically stable, then their transfer functions' poles lie in the open left half-plane, i.e., $\lambda_k,\tlambda_j\in \C^+$, and the transfer functions belong to the Hardy-Hilbert space $H_2(\C^+)$, defined as \begin{equation}
    \label{eq:H2cont}
H_2(\C^{+}) :=
\left\{
\bH \in \frak{A}(\C^{+}) \; : \;
\frac{1}{2\pi}\int_{-\infty}^{\infty}
|\bG(\iu\omega)|^2\,\du\omega < \infty
\right\}.
\end{equation}
Similarly to~\eqref{eq:ROMspace}, we use $\frak{G}_{\zeta}^n, \ \frak{G}_{\zeta}^r \subset H_2(\C^{+})$ to denote the spaces containing the continuous-time FOM and ROM transfer functions, respectively. To be precise, for any $n \in \N$ we define
\begin{multline*}
\frak{G}_{\zeta}^n :=  \left\{ \bG \in \myspan\left\{ \zeta_1,\ldots,\zeta_n\right\} : \right. \\ \left. \blambda \in (\C^+)^n,\, \lambda_j \neq \lambda_k \text{ for } j \neq k \right\}.
\end{multline*}
The hyperbolic Newton/hybrid iteration developed in Section~\ref{subsec:rgd} is formulated on the open unit disk
$\D$ and therefore naturally applies to discrete-time transfer functions in $H_2(\D)$.
This does not restrict applicability to continuous-time LTI systems, since $\D$ is conformally equivalent
to an open half-plane via a Möbius (Cayley) transform.

For $s \in \C^{+}$ and $z \in \D$ define
\begin{equation}
\label{eq:cayley}
\phi(s) := \frac{s-1}{s+1},
\quad
\phi^{-1}(z) := \frac{1+z}{1-z}.
\end{equation}
The map $\phi$ sends the imaginary axis to the unit circle $\T$ and the open right half-plane to $\D$.
%Hence, points in $\D$ correspond to (mirror-image) points in $\C^{+}$.
%To transport the $H_2$ inner product exactly, we use the standard weighted composition operator induced by
We use the standard weighted composition operator induced by the conformal map $\cU : H_2(\C^{+}) \to H_2(\D)$ given by
\begin{equation}
\label{eq:unitary}
\cU[\bH](z)
:=
\frac{\sqrt{2}}{1-\phi^{-1}(z)}\,
\bH\left(\phi^{-1}(z)\right).
\end{equation}
The mapping $\cU$ is unitary, i.e.,
%${\| \cU [\bH]\|}_{H_2(\D)} = {\|\bH\|}_{H_2(\C^{+})}$ and
${\langle  \cU [\bF],  \cU [\bH]\rangle}_{H_2(\D)}={\langle \bF,\bH\rangle}_{H_2(\C^{+})}$.
The inverse mapping is given by
\begin{equation}
\label{eq:unitary_inv}
\cU^{-1}[ \bG](s)
=
\frac{\sqrt{2}}{\phi(s)+1}\,
 \bG\!\left(\phi(s)\right).
\end{equation}

In conclusion, given a continuous-time transfer function $\bH \in \frak{G}_{\zeta}^n \subset H_2(\C^{+})$, we set $\bG := \cU[\bH] \in \fomspace \subset H_2(\D)$,
apply the disk-based hyperbolic Newton/hybrid iteration (Section~\ref{subsec:rgd}) to compute an asymptotically stable discrete-time ROM
$ \bG_r \in \romspace$, and finally map back $\bH_r := \cU^{-1}[ \bG_r] \in \frak{G}_\zeta^r$.
Since $\phi^{-1}(\D)=\C^{+}$, disk-feasibility of the iterates implies that the associated mirror-image points stay in $\C^{+}$,
and therefore the resulting continuous-time reduced poles lie in the open left complex half-plane,
%in $\C^{-}$,
i.e., the model realized by $\bH_r$ is asymptotically stable.

%\section{experiments:}
%\label{sec:exp}
%\subsection{Stability, IRKA vs Newton-type}
%\subsection{EXTRA: Stability fail IRKA %vs stable versions + backtracking IRKA}

\section{Numerical experiments}
\label{sec:exp}

We now demonstrate the effectiveness of the proposed MOR methods on several benchmark examples. %on the \texttt{build} example with $n=48$ (data available at~\cite{slicot})  and a discrete-time heat equation model with $n=2000$ detailed in \cite{benchmark}.  %with FOM orders $n=48$ and $n=2000$, respectively.
%All results are reproducible using the provided implementation (see code availability statement).
The $H_2$-error ${\|\bG - \bG_r\|}_{H_2(\Omega)}$ is used as performance metric and is recorded at each iteration whenever a stable ROM of order $r$ is produced. Here, we choose $\Omega = \D$ or $\Omega = \C^{+}$ depending on whether the system is discrete- or continuous-time. MOR is performed using three Newton method variants that rely on different update rules: Euclidean Newton (EN) using~\eqref{eq:NewtonStep}, hyperbolic Newton (HN) using~\eqref{eq:GeodesicNewton}, and the proposed hybrid Newton iteration (hHN) defined in~\eqref{eq:hybrid}. In addition, discrete- or continuous-time IRKA~\cite{gugercin2008h_2} is used as a baseline method.

We begin with a discrete-time experiment obtained via MATLAB's \texttt{c2d} discretization applied to the \texttt{build} benchmark with $n=48$~\cite{slicot,benchmark}. 
%Model reduction is performed using four methods: Euclidean Newton (EN), discrete-time IRKA in~\cite{gugercin2008h_2}, and the proposed hyperbolic Newton (HN) and hybrid Newton iteration (hHN). 
The experiments are conducted for reduced orders $r \in \{2,5,10\}$, and all methods are initialized from the same mirror image poles. The results, shown in Fig.~\ref{fig:multiple_r}, illustrate the $H_2(\D)$-error convergence across iterations. We find that the proposed hHN iteration matches the performance of EN, while ensuring stability of the ROM at every iteration.

Next, we consider the continuous-time version of the \texttt{build} benchmark. Here, continuous-time IRKA is used as a baseline and compared with the proposed Newton-based framework after applying the unitary transformation in Section~\ref{sec:continuous}. The comparison is performed for $r=10$, and the results are shown in Fig.~\ref{fig:compare_newton_cont} (top). In this case, IRKA exhibits persistent instability during the iteration, whereas the proposed methods achieve $H_2(\C^+)$-errors comparable to EN, while maintaining stability throughout the iteration.

Finally, we consider a large-scale heat equation example with FOM order $n=2000$, based on a spatial finite-difference discretization of the PDE
\begin{equation*}
\frac{\partial \psi(v,t)}{\partial t}
= \frac{\partial^2 \psi(v,t)}{\partial v^2}
+ \sigma \psi(v,t),
\quad v \in (0,1), \ t > 0.
\end{equation*}
The system is interpreted as an LTI control system with boundary input $\bu(t)=\psi(0,t)$ and output $\by(t)=\psi (1,t)$, representing heat injection at one end of a rod and temperature measurement at the other. The discretization order corresponds to the number of equidistantly spaced spatial grid points.
The parameter $\sigma$ shifts the system poles towards the imaginary axis, making the problem increasingly challenging and prone to instability for classical methods such as IRKA. The results for this experiment, shown in Fig.~\ref{fig:compare_newton_cont} (bottom) for $r=15$, highlight a different trade-off: HN remains stable but converges slowly due to the poor hyperbolic updates while EN converges faster but may produce unstable ROMs. 

\begin{figure}[h]
    \centering
    \includegraphics[width=0.51\textwidth]{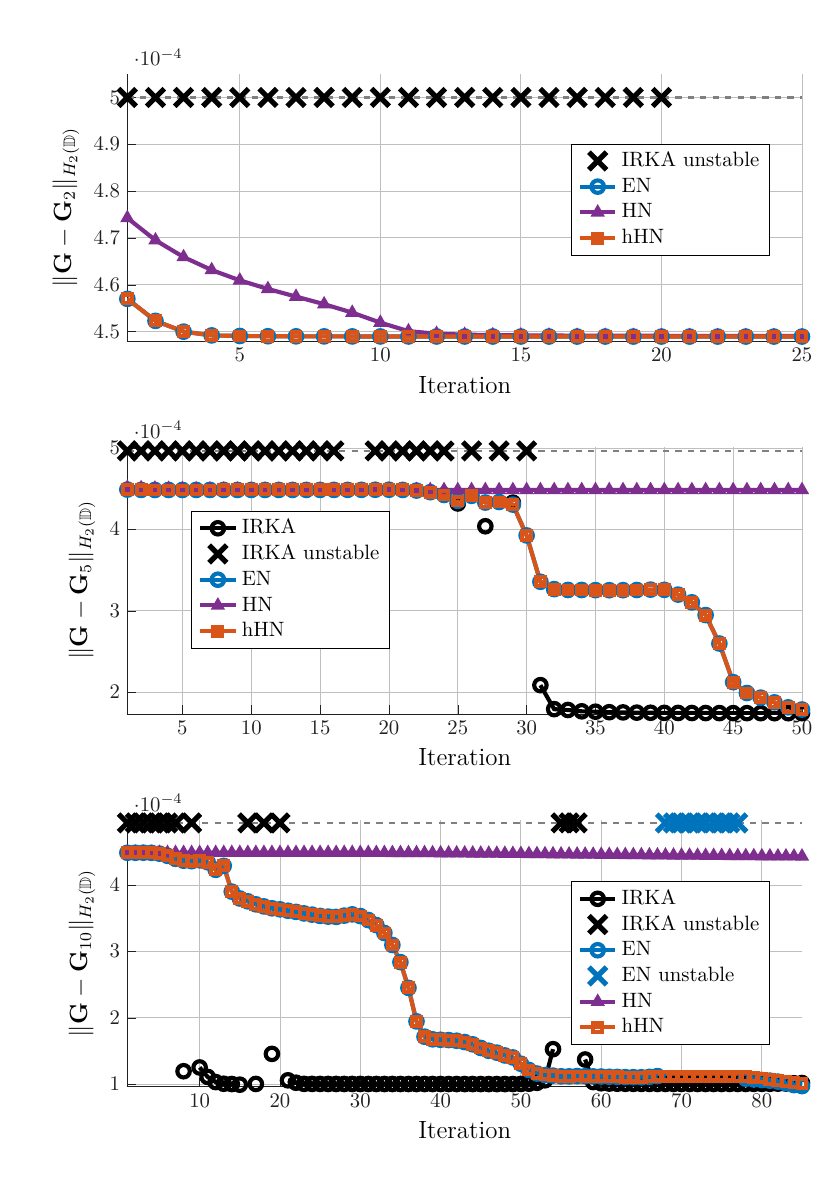}
    \caption{$H_2(\D)$-error convergence across different solvers for the discrete-time \texttt{build} benchmark ($n=48$). The sub-figures correspond to reduced orders $r \in \{2,5,10\} $. Unstable ROMs are marked with crosses.} 
    \label{fig:multiple_r}
\end{figure}

\begin{figure}[h]
    \centering
    \includegraphics[width=0.5\textwidth]{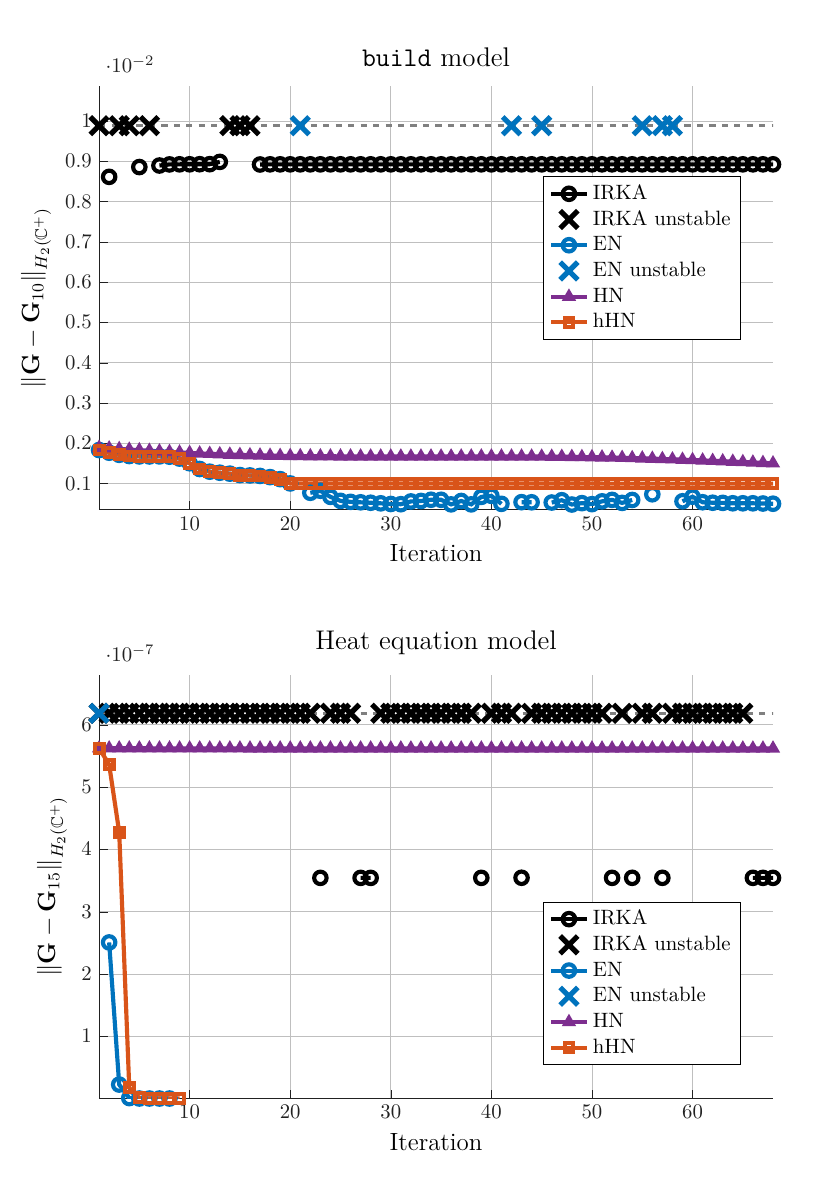}
    \caption{$H_2(\C^+)$-error convergence for the continuous-time \texttt{build} benchmark ($r=10$, top) and the heat equation benchmark ($n=2000$, $r=15$, bottom). Unstable ROMs are marked with crosses.}
    \label{fig:compare_newton_cont}
\end{figure}

In conclusion we observe that the hybrid method hHN consistently provides the best compromise, achieving fast convergence while guaranteeing stability at every iteration. Across all experiments, it emerges as the most reliable and effective approach for computing stable and accurate reduced-order models.

\section{Conclusion and future plans}
\label{sec:conc}
In this work, we proposed a novel framework for the MOR of SISO LTI systems. We introduced a new necessary condition of optimal approximation in $H_2(\D)$ and developed two novel iterative methods based on a hyperbolic Newton update. We plan to extend the developed framework to multi-input multi-output systems and aim to develop a state space representation of the equilibrium function and corresponding MOR algorithms. This would allow us to use the method on large-scale state space systems which are much more common in practice, than the pole-residue representation of the FOM assumed in this study. Finally, we plan to rigorously study the convergence properties of the proposed iteration and introduce new initialization techniques. 

\section*{Code availability}
\begin{center}
\url{https://gitlab.com/tamasdzs/mor-siso}
\end{center}

\section*{CreDit Author statement}
\textbf{Andrea Angino:} Methodology, Software, Visualization,  Writing -- Original Draft;
\textbf{Tam\'as D\'ozsa:} Conceptualization, Methodology, Software, Formal analysis, Writing -- Original draft;
\textbf{Matthias Voigt:} Investigation, Writing -- Review \& Editing, Supervision, Project administration, Funding acquisition.

\section*{Acknowledgment}
TD in part conducted his research at FernUni Schweiz, Faculty of Mathematics and Computer Scinence, where he was supported by the Swiss
Government Excellence Scholarship No. 2025.0057.

\bibliographystyle{IEEEtran}
\bibliography{refs}

@article{djukic2012dynamic,
  title={Dynamic model reduction: {A}n overview of available techniques with application to power systems},
  author={Duki{\'c}, S. D. and Sari{\'c}, A. T.},
  journal={Serb. J. Electr. Engrg.},
  volume={9},
  number={2},
  pages={131--169},
  year={2012}
}

@article{gugercin2004survey,
  title={A survey of model reduction by balanced truncation and some new results},
  author={Gugercin, S. and Antoulas, A. C.},
  journal={Internat. J. Control},
  volume={77},
  number={8},
  pages={748--766},
  year={2004},
  publisher={Taylor \& Francis}
}

@article{gugercin2008h_2,
  title={{$H_2$} model reduction for large-scale linear dynamical systems},
  author={Gugercin, S. and Antoulas, A. C. and Beattie, C.},
  journal={SIAM J. Matrix Anal. Appl.},
  volume={30},
  number={2},
  pages={609--638},
  year={2008},
  publisher={SIAM}
}

@article{FLAGG2012688,
title = {Convergence of the {I}terative {R}ational {K}rylov {A}lgorithm},
journal = {Systems  Control Lett.},
volume = {61},
number = {6},
pages = {688--691},
year = {2012},
doi = {https://doi.org/10.1016/j.sysconle.2012.03.005},
author = {Flagg, G. and Beattie, C. and Gugercin, S.}
}

@ARTICLE{RiemIRKA,
  author={Mlinarić, P. and Beattie, C. A. and Drmač, Z. and Gugercin, S.},
  journal={IEEE Trans. Automat. Control}, 
  title={I{R}{K}{A} is a {R}iemannian {G}radient {D}escent {M}ethod}, 
  year={2025},
  volume={70},
  number={5},
  pages={2979--2991},
  doi={10.1109/TAC.2024.3489416}
  }

@article{soumelidis2017hyperbolic,
  title={Hyperbolic geometrical approach to model reduction},
  author={Soumelidis, A. and Bokor, J. and Schipp, F. and Szab{\'o}, Z.},
  journal={IFAC-PapersOnLine},
  volume={50},
  number={1},
  pages={12905--12910},
  year={2017},
  publisher={Elsevier}
}

@ARTICLE{Meier,
  author={Meier, L. and Luenberger, D.},
  journal={IEEE Trans. Automat. Control}, 
  title={Approximation of linear constant systems}, 
  year={1967},
  volume={12},
  number={5},
  pages={585--588},
  doi={10.1109/TAC.1967.1098680}
}

@book{szidar,
  title={Linear Systems Theory},
  author={Szidarovszky, F. and Bahill, A. T.},
  year={1998},
  publisher={Routledge},
  address = {New York, NY},
  edition = {2nd}
}

@book{feedback,
  title={Feedback Systems: An Introduction for Scientists and Engineers},
  author={{\AA}str{\"o}m, K. J. and Murray, R. M.},
  year={2021},
  edition = {2nd},
  publisher={Princeton {U}niversity {P}ress},
  address = {Princeton, NJ}
}

@MISC{slicot,
 title =        {{SLICOT} -- {S}ubroutine {L}ibrary in {S}ystems and {C}ontrol {T}heory},
 howpublished =      {\url{http://www.slicot.org}},
 note = {Accessed: 2025-09-19}
}

@techreport{benchmark,
title = {A Collection of Benchmark Examples for Model Reduction of Linear Time Invariant Dynamical Systems},
author = {Y.~Chahlaoui and P.~{Van Dooren}},
type = {SLICOT Working Note},
number = {2002-2}, 
year = {2002},
}

@ARTICLE{SzegoProof,
  author={Ninness, B. and Gustafsson, F.},
  journal={IEEE Trans. Automat. Control}, 
  title={A unifying construction of orthonormal bases for system identification}, 
  year={1997},
  volume={42},
  number={4},
  pages={515--521},
  doi={10.1109/9.566661}
}

@article{vandenhof,
  title={Modelling and identification with rational orthogonal basis functions},
  author={Van den Hof, P. and Wahlberg, B. and Heuberger, P. and Ninness, B. and Bokor, J. and e Silva, T. O.},
  journal={IFAC Proc. Vol.},
  volume={33},
  number={15},
  pages={445--455},
  year={2000},
  publisher={Elsevier},
  doi={10.1016/S1474-6670(17)39791-4}
}

@article{dozsa2023inverses,
  title={Inverses of rational functions},
  author={D{\'o}zsa, T.},
  journal={Acta Cybernet.},
  volume={26},
  number={1},
  pages={53--69},
  year={2023}
}

@book{absil2008matrix,
  title={Optimization Algorithms on Matrix Manifolds},
  author={Absil, P.-A. and Mahony, R. and Sepulchre, R.},
  year={2008},
  publisher={Princeton University Press},
  address={Princeton, NJ},
  isbn={978-0-691-13298-3}
}

@book{boumal2023intro,
  title={An Introduction to Optimization on Smooth Manifolds},
  author={Boumal, N.},
  year={2023},
  publisher={Cambridge University Press},
  address = {Camdridge},
  doi={10.1017/9781009166164}
}

@book{conn2000trust,
  title     = {Trust Region Methods},
  author    = {Conn, Andrew R. and Gould, Nicholas I. M. and Toint, Philippe L.},
  year      = {2000},
  publisher = {SIAM},
  address   = {Philadelphia, PA},
  isbn      = {9780898714609}
}

@book{henrici1988applied,
  title={Applied and Computational Complex Analysis Volume~I: Power Series Integration Conformal Mapping Location of Zeros},
  author={Henrici, P.},
  year={1988},
  series = {Wiley Classics Lib.},
  publisher={Wiley},
  address = {Hoboken, NJ}
}

\end{document}